\RequirePackage{fix-cm}
\documentclass[smallextended]{svjour3}       
\smartqed  
\usepackage{graphicx}
\usepackage{amsmath}
\usepackage{amssymb}
\usepackage{natbib}
\newtheorem{Theorem}{Theorem}
\newtheorem{Lemma}{Lemma}

\newtheorem{Proposition}{Proposition}
\newtheorem{Corollary}{Corollary}

\newcommand{\x}{X}
\newcommand{\s}{\mathbb{S}}
\newcommand{\p}{\mathbb{P}}

\newcommand{\cH}{\mathcal{H}}
\newcommand{\cP}{\mathcal{P}}
\newcommand{\cS}{\mathcal{S}}
\newcommand{\cV}{\mathcal{V}}
\newcommand{\cW}{\mathcal{W}}
\newcommand{\cE}{\mathcal{E}}

\begin{document}

\title{Species notions that combine phylogenetic trees and phenotypic partitions}


\author{Anica Hoppe \and Sonja T\"urpitz \and Mike Steel}


\institute{Anica Hoppe, Sonja T\"urpitz \at
          Institute of Mathematics and Computer Science, Ernst-Moritz-Arndt University,  Greifswald, Germany.  \\
           \and
           M. Steel \at
           Biomathematics Research Centre, University of Canterbury, Christchurch, NZ\\
            \email{mike.steel@canterbury.ac.nz}  
}

\date{Received: date / Accepted: date}

\maketitle

\begin{abstract}
	A recent paper \citep{manceau2016species} developed a novel approach for describing  two well-defined notions of  `species'  based on a  phylogenetic tree and a phenotypic partition.
In this paper, we explore some further combinatorial properties of this approach and describe  an extension that allows an arbitrary number of phenotypic partitions to be combined with a phylogenetic tree for these two species notions.

\end{abstract}

\newpage

	\section{Introduction}
	In biological classification, there have been numerous attempts to define notions of `species', from Aristotle through to the present \citep{deq, que}.  Today, phylogenetic trees and networks provide a framework for addressing this question \citep{bau, bau2, noa}; however, phylogeny alone captures only part of the concept of species. For example, two taxa may appear as different leaves in a tree (e.g. because of genetic differences in a neutral gene) but still be virtually identical (and perhaps even able to interbreed to produce viable offspring in the case where the taxa are diploid) and so would not be considered as belonging to different species.  Thus, phenotypical characters also play a role in the concept of species. 
	
	However, phenotypical characters alone are also not sufficient to delineate species, because events such as convergent evolution can result in  different species appearing quite similar. For example,  two species separated by a long period of evolution may share many morphological traits because they have evolved under similar selection pressures \citep{hol}.
	
If we regard species classification as the construction of a partition of a set $X$ of taxa (e.g. organisms, populations, etc.)  into disjoint sets (`species'), then a natural question arises: How can one construct such a partition in a canonical way so as to satisfy various desirable criteria by using a phylogenetic tree, and one or more phenotypic characters?  A novel approach to this question was explored recently by \cite{manceau2016species}, and provides the motivation and focus for this paper.

Suppose we are given a rooted binary phylogenetic tree $T$ on a sample $\x$ of individuals (or, equivalently, as a `hierarchy' on $\x$, defined shortly), together with a partition $\cP$ of $\x$, called  a phenotypic partition (such a partition is induced by a phenotypic character by grouping taxa in the same state together).  The desired output is a `species partition' $\cS$ of $\x$ that satisfies three desirable properties (defined shortly)  which should be satisfied  by a classification into species, namely, {\em heterotypy} between species, {\em homotypy} within species and {\em exclusivity} of each species. 
They showed that all three properties cannot generally be simultaneously satisfied; however, any pair of them can be, in which case there are two canonical constructions of a species partition $\cS$ of $\x$. Related combinatorial notions were also considered in a somewhat different setting in \cite{ald}; see also \cite{ale}, \cite{dre}, \cite{kor} and \cite{kwo} for combinatorial approaches to the species problem.

In this paper, we generalize some of the main results from \cite{manceau2016species}, by allowing an arbitrary number of phenotypic partitions, and by 
lifting the requirement that the input tree $T$ be binary. We also provide a short proof of a key result from that paper by using lattice theory, and, in the final section, establish conditions under which phenotypic partitions
can be realized by state changes (indicated by `marks') on edges of the tree $T$.

	\subsection{Preliminaries}
We first review some basic notions from phylogenetic combinatorics  (further details can be found in \cite{drebook} or \cite{steel2016phylogeny}).

	A \textit{rooted phylogenetic tree} $T$ on leaf set $\x$  is a rooted tree in which each non-leaf vertex (including the root) is unlabelled and has at least two outgoing edges.   Given any rooted phylogenetic tree $T$, there is a partial order $\preceq_T$ on the vertices of $T$, defined by $u\preceq_T v$ if $u=v$ or if $u$ lies on the path from the root to $v$. Given a rooted tree $T=(V,E)$ and a non-empty subset $C$ of $V$, there is a (unique) vertex $v$ for which $v\preceq_T c$ for all $c\in C$ and which is maximal (under $\preceq_T$) with respect to this property; this vertex $v$ is called the \textit{lowest common ancestor (lca)} of $C$ in $T$ and is denoted $\text{lca}_T(C)$.
	
	A \textit{hierarchy} $\cH$ is a collection of non-empty subsets of $\x$ that contains $\x$ and that satisfies the following {\em nesting} property:
 $$\mbox{For all }   A,B\in \cH,  A\cap B\in \{\emptyset,A,B\}.$$
If a hierarchy also has the property that the singleton $\{x\}$ is an element of $\cH$ for every $x\in\x$  then we say that $\cH$  is a \textit{hierarchy on $\x$}.
	
	An element of a hierarchy $\cH$ is also called a \textit{cluster} (it is sometimes also referred to in the literature as a `clade').   Given a hierarchy on $\x$, there is a unique rooted phylogenetic tree $T$ on $\x$ with the property that for each cluster there is a unique vertex in $T$ such that the cluster consists of the labels of exactly the leaves descending from this vertex. In this way, hierarchies on $\x$ and rooted phylogenetic trees on leaf set $\x$ are essentially equivalent and will be used interchangeably.
		
	A \textit{partition} of a set $\x$ is a division of $\x$ into a set of non-empty subsets, called {\em blocks},  such that every element of $\x$ is in exactly one block. Let $\p(\x)$ denote the set of all partitions of $\x$. A \textit{phenotypic partition} $\cP$ is a partition of $\x$, such that the blocks of the partition correspond to the different `phenotypes' of the individuals in $\x$ according to some biological characteristics. For example, birds can have the phenotype `able to fly' or `not able to fly', resulting in a phenotypic partition with two blocks. 
	
	A \textit{species partition} $\cS$ is a partition of $\x$ in which the blocks are called `species'.   We are particularly interested in species partitions that satisfy one or more of the following three desirable properties from \cite{manceau2016species}:	
	
	\begin{enumerate} 
		\item [(A)] \textit{Heterotypy between species.} Individuals in different species are phenotypically different.  Formally, for each phenotype $P\in\cP$ and for each species $S\in \cS$, either $P\subseteq S$ or $P\cap S = \emptyset$. 
		\item [(B)] \textit{Homotypy within species.} Individuals in the same species are phenotypically identical. Formally,  for each phenotype $P\in \cP$ and for each species $S\in \cS$, either $S\subseteq P$ or $P\cap S=\emptyset$.
		\item[(E)] \textit{Exclusivity.} All species are exclusive. Formally, each species is a cluster of $T$ (i.e. $\cS\subseteq\cH$).
	\end{enumerate}
	
	For any species partition  $\cS$ of $\x$ satisfying property (E) for a given hierarchy $\cH$  on $\x$, there is a corresponding hierarchy $\cH_{\cS}$ which is defined in \cite{manceau2016species} by
	\begin{equation}
	\label{hierar}
	\cH_{\cS}:=\{H\in\cH :   \mbox{ there exists some } S\in\cS \mbox{ with } S\subseteq H\}.
	\end{equation}

\subsection{The lattice of partitions of $\x$}
\label{partitions}

	Since we want to apply combinatorial arguments, we first recall some terms from lattice theory.
	Let $(Y,\preceq)$ be a \textit{partially  ordered set} (poset) consisting of a set $Y$ and a partial order $\preceq$. Given a set $A\subseteq Y$ with $A\neq \emptyset$, we say that $p\in Y$ is an \textit{upper bound} for $A$ if $a\preceq p$ for all $a\in A$. The \textit{least upper bound} for $A$ (${\rm lub}(A)$) is an upper bound such that for any other upper bound $p$ one has ${\rm lub}(A)\preceq p$. The \textit{lower bound} and \textit{greatest lower bound} (${\rm glb}(A)$) are defined analogously. A \textit{lattice} is a partially ordered set  $(Y,\preceq)$ with the property that for all non-empty subsets  $A$ of $Y$, ${\rm lub}(A)$ and ${\rm glb}(A)$ exist and both are unique (for further background and details, see \cite{bona2011walk}).

Next, we recall some important properties of the poset $\p(\x)$ (of partitions of $\x$), beginning with  notion of \textit{refinement} of partitions.
Let $\cS$ and $\cS'$ be two partitions of the set $X$. We say that $\cS$ is \textit{finer} than $\cS'$ and $\cS'$ is\textit{ coarser} than $\cS$, denoted $\cS\preceq\cS'$, if for each $S\in \cS$ and each $S'\in\cS'$, either $S\subseteq S'$ or $S\cap S'=\emptyset$. The  relation $\preceq$ is a partial order but  not a linear order (i.e.  two partitions $\cS$ and $\cS'$ of the set $\x$ cannot always be compared; for example, consider $\cS=\{\{a,b\},\{c\},\{d\}\}$ and $\cS'=\{\{a\},\{b\},\{c,d\}\}$).

It can be shown that the set of all partitions of the set $\x$ ordered by refinement (i.e.  $(\p(\x),\preceq)$)  is a lattice. In other words, the ${\rm lub}$ and the ${\rm glb}$ of any subset of $\p(\x)$ exist \citep{bona2011walk}. We describe the ${\rm glb}$ and ${\rm lub}$ explicitly now. Given $k$ partitions $\Sigma_1,\ldots,\Sigma_k\in\p(\x)$, their ${\rm glb}$ is the set of non-empty intersections of blocks each chosen from a different partition. That is, ${\rm glb}(\Sigma_1,\ldots,\Sigma_k)=\{\bigcap\limits_{i=1}^{k} B^{(i)}: B^{(i)}\in\Sigma_i\}-\{\emptyset\}$.
 To define the ${\rm lub}$,  first define the relation $\sim$ by $x\sim y$ if there exists $i\in\{1,\ldots,k\}$ such that $\{x,y\}\subseteq B\in\Sigma_i$. In other words, $x\sim y$ if $x$ and $y$ are in the same block in at least one of the $k$ partitions of $\x$. Let $\approx$ be the transitive closure of $\sim$.
The ${\rm lub}$ of $\Sigma_1,\ldots,\Sigma_k$ is then defined as the set of equivalence classes of $\x$ under $\approx$.
Note,  that $$\cS_0= \{\{x\}: x \in \x\} \mbox{ and } \cS_1= \{\x\}$$ are the minimal and maximal elements of $\p(\x)$, respectively, under the partial order $\preceq$.

	\section{Discussion of earlier results}\label{earlier}
In this section, we describe some of the main results of \cite{manceau2016species}, which we have organised into two theorems. The first theorem has a short proof, as presented in that
paper, the second theorem involved a more complicated argument, and we present a  shorter lattice-theoretic proof.
From now on, $\cH$ always denotes a hierarchy of $\x$, $\cP$ always denotes a phenotypic partition, and $\cS$ always denotes a species partition.

\begin{Theorem} [\cite{manceau2016species}]
\label{mainman1}
Given a phenotypic partition $\cP$ of $\x$ and a hierarchy $\cH$ on  $\x$:
\begin{itemize}
\item[(i)]
$\cS$ satisfies property  (A) if and only if $\cP\preceq\cS$, and $\cS$ satisfies property (B) if and only if $\cS\preceq\cP$. In particular, if  $\cS$ satisfies properties (AB), then $\cS=\cP$.
\item[(ii)]Unless $\cP$ and $\cH$ satisfy  $\cP\subseteq\cH$ (i.e.  each phenotype is a cluster), there is no species partition $\cS$ that satisfies all three properties (ABE).
\item[(iii)] Any \underline{two}   properties from  (A), (B) and (E) can be satisfied by at least one species partition $\cS$. Specifically, $\cS=\cP$ satisfies (AB), $\cS_1$ satisfies (AE), and $\cS_0$ satisfies (BE).
\end{itemize}
\end{Theorem}

Theorem~\ref{mainman1} raises an interesting question.  One can satisfy (AE) by the very coarsest partition $\cS_1$, and (BE) by the very finest partition $\cS_0$. However, neither of  these is particularly relevant for biology.  Rather, one would like to find a finest partition satisfying (AE) and a coarsest partition satisfying (BE). Ideally, such partitions should also be uniquely determined by those properties.  Fortunately, this turns out to be the case.  A second main theorem  from \cite{manceau2016species} is the following.

\begin{Theorem} [\cite{manceau2016species}]
\label{mainman2}
Given $\cP$ and $\cH$:
\begin{itemize}
\item[(i)] There is a  unique finest partition  of $\x$ satisfying (AE) (heterotypy between species and exclusivity). 
\item[(ii)] There is a unique coarsest partition of $\x$ satisfying (BE) (homotypy within species and exclusivity).
\end{itemize}
\end{Theorem}
Following \cite{manceau2016species}, we will refer to the  unique species partitions referred to in parts  (i) and (ii) of Theorem~\ref{mainman2} as the  {\em loose} and {\em lacy} species partitions, respectively, and denote them  as $\cS_{\rm loose}$ ($=\cS_{\rm loose}(\cH, \cP))$ and 
$\cS_{\rm lacy}$ ($=\cS_{\rm lacy}(\cH, \cP))$, respectively.

 We will describe example of these two partitions in the next section, but first we provide a short lattice-theoretic proof of Theorem~\ref{mainman2}. This proof relies on the following lemma.

\begin{Lemma}\label{lemma}
	Given  a non-empty subset $\p$ of partitions of $\x$ (i.e. $\p\subseteq \p(X)$) and a hierarchy $\cH$ on $\x$, if $\Sigma \subseteq\cH$ for each $\Sigma \in \p$ then ${\rm glb}(\p)\subseteq\cH$ and ${\rm lub}(\p)\subseteq\cH$.
\end{Lemma}

\begin{proof}
	Consider ${\rm glb}(\p)$. Because $\Sigma \subseteq\cH$ for each $\Sigma \in \p$, and $\cH$ satisfies  the nesting property of a hierarchy,  if we select a  block from each partition in $\p$, then the resulting sets $B_1, B_2, \ldots, B_{|\p|}$ satisfy $\bigcap\limits_{i=1}^{|\p|} B_i \in \{B_1, \ldots,B_{|\p|},\emptyset\}$. Hence, all elements of ${\rm glb}(\p)$ are also elements of $\cH$ and thus, ${\rm glb}(\p)\subseteq\cH$.  Next consider ${\rm lub}(\p)$. Because $\Sigma\subseteq\cH$ for all $\Sigma \in \p$, the relation $\sim$ is already transitive. To see this, observe that if two blocks $B_1$ and $B_2$ of two different  partitions from $\p$ share at least one element, then the nesting property of hierarchies implies that  one of those blocks ($B_1$ or $B_2$)  contains the other block. Consequently, if $x\sim y$ and $y\sim z$ with $x,y\in B_1$ and $y,z\in B_2$, then either $x\in B_2$ or $z\in B_1$, and thus, $x\sim z$. Thus, $\sim$ is transitive and so is equal to the equivalence relation $\approx$ whose equivalence classes are the blocks of ${\rm lub}(\p)$. Consequently, each block of the partition ${\rm lub}(\p)$ is  also an element of $\cH$ and thus, ${\rm lub}(\p)\subseteq\cH$. \end{proof}

We can now provide a short proof of Theorem~\ref{mainman2}.

\begin{proof}
Let \[\s_{AE}:=\{\Sigma\in\p(\x):\Sigma\subseteq\cH,\cP\preceq\Sigma\},\text{ and } \s_{BE}:=\{\Sigma\in\p(\x):\Sigma\subseteq\cH,\Sigma\preceq\cP\}.\]
Since $\s_{AE},\s_{BE}\subseteq\p(\x)$, and $(\p(\x),\preceq)$ is  a lattice, ${\rm glb}(\s_{AE})$ and ${\rm lub}(\s_{BE})$ exist.
Therefore, to establish Theorem~\ref{mainman2}, it suffices to show that:
\begin{equation}
\label{eqman}
{\rm glb}(\s_{AE})\in\s_{AE} \mbox{ and } {\rm lub}(\s_{BE})\in\s_{BE},
\end{equation}
from which it follows that  $\cS_{\rm loose}={\rm glb}(\s_{AE}) \mbox{ and } \cS_{\rm lacy}={\rm lub}(\s_{BE}).$
We establish  Eqn.~(\ref{eqman}) as follows. 
 The two  sets  ${\rm glb}(\s_{AE})$ and ${\rm lub}(\s_{BE})$ are not empty because $\cS_1\in \s_{AE}$ and $\cS_0\in \s_{BE}$.  Since all  elements of $\s_{AE}$ and of $\s_{BE}$ satisfy  property (E) (i.e.  $\cS\subseteq\cH$ for all $\cS\in \s_{AE}, \s_{BE}$),  Lemma \ref{lemma} implies that  ${\rm glb}(\s_{AE})$ and ${\rm lub}(\s_{BE})$ satisfy property (E). Also, because elements of $\s_{AE}$ satisfy  property (A) (i.e.  $\cP\preceq\cS$ for all $\cS\in \s_{AE}$), it follows that $\cP$ is a lower bound for $\s_{AE}$. Since ${\rm glb}(\s_{AE})$ is the greatest lower bound we know that it is coarser than all other lower bounds and, in particular, $\cP\preceq {\rm glb}(\s_{AE})$. Thus, ${\rm glb}(\s_{AE})$ satisfies property (A) as well. The proof that ${\rm lub}(\s_{BE})$ satisfies property (B) is analogous. In summary,  ${\rm glb}(\s_{AE})\in\s_{AE}$ and ${\rm lub}(\s_{BE})\in\s_{BE}$.
\end{proof}


\section{Constructing the loose and lacy species partitions}\label{corrected}
We now describe an explicit and readily computable way to construct the loose and lacy species partitions (our treatment differs in some details from the construction proposed by \cite{manceau2016species}).   The process involves first constructing subsets $\cH_1$ and $\cH_2$ of $\cH$. The loose and lacy species partitions are then the maximal elements (under set inclusion) of $\cH_1$ and $\cH_2$  respectively.
We first introduce some additional notation.  For $P \in \cP$, let
\begin{equation}
\label{hpeq}
h_P = \bigcap_{h \in \cH : P \subseteq h} h.
\end{equation}
Notice that $P \subseteq h_P \in \cH$, and that $h_P$ is the unique minimal cluster in $\cH$ that contains $P$ (it is the cluster in $\cH$ that corresponds to ${\rm lca}_T(P)$ where $T$ is the phylogenetic $X$--tree associated with $\cH$). 
It is  also possible that $h_P = h_Q$ for different blocks $P$ and $Q$ of  $ \cP$.

Next, given $\cP$ and $\cH$, define the following two sets:
$$\cH_1 := \{h_P : P \in \cP\}, \mbox{ and }
\cH_2 := \{h \in \cH : \exists P \in \cP: h \subseteq P\}.$$
Observe that $\cH_1$ and $\cH_2$ are both subsets of $\cH$. 
The following proposition provides an explicit description of the loose and lacy partitions; its proof is provided in the Appendix.

\begin{Proposition}
\label{construct}
Given a phenotypic partition $\cP$ of $X$ and a hierarchy  $\cH$ on $X$,
$$\cS_{\rm loose} \mbox{ is the set of maximal  elements of } \cH_1;$$
$$ \cS_{\rm lacy} \mbox{ is the set of maximal elements of } \cH_2.$$
\end{Proposition}


In the phylogenetic tree $T$ corresponding to the hierarchy $\cH$ the elements in $\cH_1$ are the lowest common ancestors in $T$ of each element in $\cP$. The loose species partition  contains, for each leaf of $T$, the first last common ancestor which is in $\cH_1$ that lies on the path from the root to the leaf. \\
The elements of $\cH_2$  correspond to the vertices of the phylogenetic tree $T$ associated with $\cH$  for which all descended leaves have the same phenotype. The lacy species partition  contains, for each leaf of $T$, the first element which is in  $\cH_2$ that lies on the path from the root to the leaf.

An example (based on Fig. 5 from \cite{manceau2016species}) is shown in Fig.~\ref{manfig}, for the set $\x=\{1,2,3,4,5,6,7,8,9\}$  together with the hierarchy $\cH$ on $X$ defined by
 $$\cH=\{\{x\}: x\in \x\} \cup \{\{2,3\}, \{1,2,3\}, \{4,5\},\{6,7\},\{8,9\},\{1,2,3,4,5\},\{6,7,8,9\},\x\}$$ along with the phenotypic partition
$\cP=\{\{1,2,3\},\{4\},\{5\},\{6,8,9\},\{7\}\}.$
In this example,
$$\cH_1= \{\{1,2,3\}, \{4\}, \{5\}, \{7\}, \{6,7,8,9\}\}, \mbox{ and so}$$
$$\cS_{\rm loose}= \{\{1,2,3\},\{4\},\{5\},\{6,7,8,9\}\}.$$
In addition, 
$$\cH_2=\{\{x\}: x\in \x\} \cup \{\{2,3\}, \{1,2,3\}, \{8,9\}\}, \mbox{ and so }$$
$$\cS_{\rm lacy}=\{\{4\}, \{5\}, \{6\}, \{7\}, \{1,2,3\}, \{8,9\}\}.$$
\begin{figure}[h]
\centering
\includegraphics[scale=0.8]{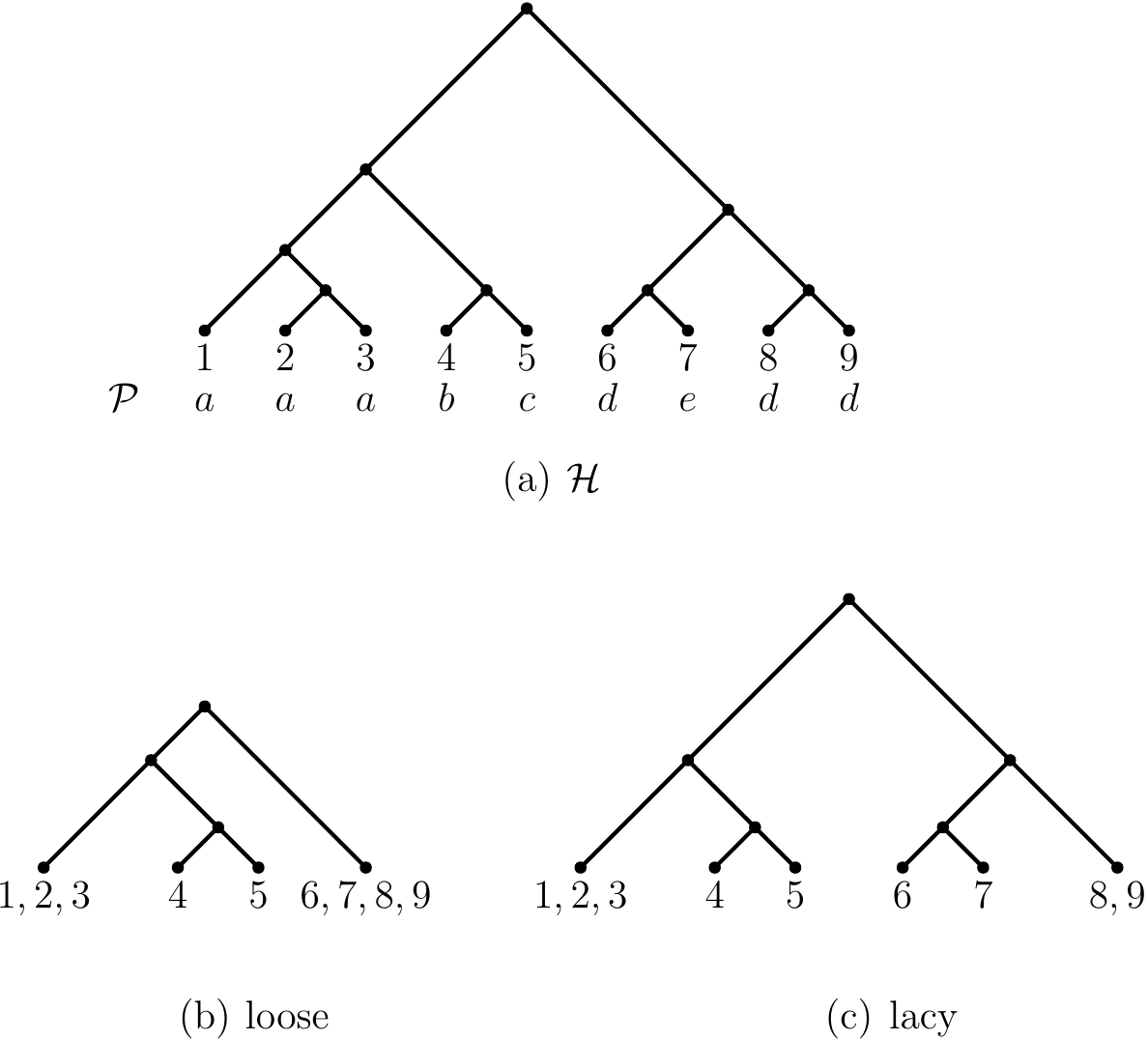}
	\caption{(a) A hierarchy $\cH$ and a phenotypic partition $\cP$ indicated by the states at the leaves. (b) The corresponding loose species partition $\cS_{\rm loose}$ and its associated hierarchy $\cH_{\rm loose}$ (c).  The corresponding lacy species partition $\cS_{\rm lacy}$  and its associated hierarchy $\cH_{\rm lacy}$. \label{exampleH}}
	\label{manfig}
\end{figure}

The corresponding induced hierarchies (given by Eqn.~(\ref{hierar})) are given by: 
$$\cH_{\rm loose}=\cS_{\rm loose} \cup \{\{1,2,3,4,5\}, \x\}
 \mbox{ and }  \cH_{\rm lacy}=\cS_{\rm lacy} \cup\{\{4,5\},\{6,7\}, \{1,2,3,4,5\},\{6,7,8,9\},\x\}.$$


	\section{Extending the theory to $k\geq 2$ phenotypic partitions}\label{k part}
Suppose now, that we have $k\geq 2$ phenotypic partitions $\cP_1,\dots,\cP_k$ with $\cP_i\in \p(\x)$ for $i=1,\dots,k$ and, as before, a hierarchy $\cH$ on $\x$.\\
Define the sets of species partitions satisfying properties (AE), (BE) as the \mbox{following}:  $$\s_{AE}^{(k)}:=\{\Sigma\in\p(\x):\Sigma\subseteq\cH,\cP_i\preceq\Sigma,\mbox{ for each } i=1,\dots,k\}, \mbox{ and }$$ $$\s_{BE}^{(k)}:=\{\Sigma\in\p(\x):\Sigma\subseteq\cH,\Sigma\preceq\cP_i, \mbox{ for each }  i=1,\dots,k\}.$$

Some natural questions arise now: Is ${\rm glb}(\s_{AE}^{(k)})\in\s_{AE}^{(k)}$? Is ${\rm lub}(\s_{BE}^{(k)})\in\s_{BE}^{(k)}$?
And if so, how are these two partitions related to $\cS_{\rm loose}$ and $\cS_{\rm lacy}$?
We will show that the case $k\geq 2$ can be reduced to the earlier case where $k=1$.

\begin{Theorem}\label{more}
Given a hierarchy $\cH$ on $\x$ and phenotypic partitions $\cP_1,\dots,\cP_k$ of $\x$, let 
$$\cP^+ := {\rm lub}(\cP_1,\dots,\cP_k) \mbox{ and } {\cP}^- := {\rm glb}(\cP_1,\dots,\cP_k).$$
	$$\mbox{ Then } {\rm glb}(\s_{AE}^{(k)})=\cS_{\rm loose}(\cH,\cP^+) \in\s_{AE}^{(k)},$$
	$$ \mbox{ and } {\rm lub}(\s_{BE}^{(k)})=\cS_{\rm lacy}(\cH,\cP^-) \in\s_{BE}^{(k)}.$$

\end{Theorem}
\begin{proof}
First, we show that:
\begin{equation}\label{SAEeq}
\s_{AE}^{(k)}=\{\Sigma\in\p(\x):\Sigma\subseteq\cH,\cP^+\preceq\Sigma\}.
\end{equation}
By the definition of $\s_{AE}^{(k)}$, it follows that $\cP_i\preceq \Sigma$ for all $\Sigma\in\s_{AE}^{(k)}$ and for all $i=1,\ldots, k$. Furthermore, for all $i=1,\dots,k$ we have $\cP_i\preceq \cP^+$. By the definition of the ${\rm lub}$, there is no $\cP\in\p(\x)$ with  $\cP\subseteq\cH$ and with $\cP$ lying strictly between $\cP_i$ and $\cP^+$ under the refinement partial order $\preceq$ for all $i=1,\dots,k$. Therefore, we obtain Eqn.~(\ref{SAEeq}).
For $\s_{BE}^{(k)}$ an analogous argument applies, using ${\rm glb}$ instead of ${\rm lub}$. Specifically: $$\s_{BE}^{(k)}=\{\Sigma\in\p(\x):\Sigma\subseteq\cH,\Sigma\preceq \cP^-\}.$$

Now, $\cP^+$ and $\cP^-$ are  partitions of $\x$, and for each of these two phenotypic partitions Eqn.~(\ref{eqman}) ensures that ${\rm glb}(\s_{AE}^{(k)})= \cS_{\rm loose}(\cH,\cP^+) \in\s_{AE}^{(k)}$
 and ${\rm lub}(\s_{BE}^{(k)})=\cS_{\rm lacy}(\cH,\cP^-) \in\s_{BE}^{(k)}$, as claimed.\\

\end{proof} 

\subsection{Combining properties for $k\geq 2$ phenotypic partitions}\label{combining}
Until now, we were looking for species partitions that are subsets of the given hierarchy and either finer or coarser than the given phenotypic partitions, in other words,  satisfying either (AE) or (BE). Another question concerning $k=2$ phenotypic partitions $\cP_1$ and $\cP_2$ is the following:  Is there a subset $\Sigma$ of $\cH$ for which $\cP_1 \preceq \Sigma \preceq \cP_2$? In other words, is there a species partition that satisfies (AE) for one phenotypic partition and (BE) for a second one?  The answer to this question is provided in the following result.

\begin{Theorem}\label{inBetween}
Given phenotypic partitions $\cP_1, \cP_2$ of $\x$  and a hierarchy $\cH$ on $\x$,   there is a subset $\Sigma$ of $\cH$ with $\cP_1 \preceq \Sigma \preceq \cP_2$ if and only if $\cS_{\rm loose}(\cH,\cP_1) \preceq \cS_{\rm lacy}(\cH,\cP_2)$.
\end{Theorem}
\begin{proof}
($\Rightarrow$) First suppose that  $\cP_1 \preceq \Sigma \preceq \cP_2$ holds for some $\Sigma \in \p(\x), \Sigma\subseteq\cH$. 
		This implies the following:
		\begin{itemize}
			\item $\cP_1 \preceq \Sigma \subseteq \cH$, so $\Sigma$ is an element of the set $\s_{AE}$ for $\cP_1$
			and $\cS_{\rm loose}(\cH, \cP_1) \preceq \Sigma$ (since $\cS_{\rm loose}$ is a lower bound to $\s_{AE}$ for $\cP_1$);
			\item $\cH \supseteq \Sigma \preceq \cP_2$, so $\Sigma$ is an element of the set $\s_{BE}$ for $\cP_2$
			and  $\Sigma \preceq \cS_{\rm lacy}(\cH,\cP_2)$ (since $\cS_{\rm lacy}$ is an upper bound to $\s_{BE}$ for $\cP_2$).
		\end{itemize}Consequently, $\cS_{\rm loose}(\cH,\cP_1) \preceq \cS_{\rm lacy}(\cH,\cP_2)$, which
 establishes the forward implication.

($\Leftarrow$) Conversely, suppose that  $\cS_{\rm loose}(\cH,\cP_1) \preceq \cS_{\rm lacy}(\cH,\cP_2)$ holds.
		Since
$$\cP_1 \preceq \cS_{\rm loose}(\cH,\cP_1) \mbox{ and } \cS_{\rm lacy}(\cH,\cP_2) \preceq \cP_2,$$
		by the definitions of $\cS_{\rm loose}$ and $\cS_{\rm lacy}$, we have $\cP_1 \preceq \cS_{\rm loose}(\cH,\cP_1) \preceq \cS_{\rm lacy}(\cH,\cP_2) \preceq \cP_2$.
		Moreover, $\cS_{\rm loose}(\cH,\cP_1), \cS_{\rm lacy}(\cH,\cP_2) \subseteq \cH$,  by definition.
		Therefore, by choosing either $\Sigma:=\cS_{\rm loose}(\cH,\cP_1)$ or $\Sigma:=\cS_{\rm lacy}(\cH,\cP_2)$, it holds that: $$\cP_1 \preceq \cS_{\rm loose}(\cH,\cP_1) \preceq \Sigma \preceq \cS_{\rm lacy}(\cH,\cP_2) \preceq \cP_2.$$
		This establishes the reverse implication.
\end{proof}

If we focus on the set of all species partitions that lie between the two given phenotypic partitions (i.e.  $\s:=\{\Sigma\subseteq\cH: \cP_1 \preceq \Sigma \preceq \cP_2\}$), then we automatically see that $\cS_{\rm loose}(\cH,\cP_1)$ is the finest element of $\s$ and  $\cS_{\rm lacy}(\cH,\cP_2)$ is the coarsest.\\

We now extend this problem to $k\geq 2$ phenotypic partitions. The question is: Given phenotypic partitions $\cP_1,\ldots,\cP_k$ and an integer $l$ with $0 \leq l \leq k$ is there a subset  $\Sigma$ of $\cH$ for which $\cP_1,\ldots,\cP_l \preceq \Sigma \preceq \cP_{l+1},\ldots,\cP_k$? The answer to this question is provided as follows. 

\begin{Corollary}
	There exists a $\Sigma\subseteq\cH$ with $\cP_1,\ldots,\cP_l \preceq \Sigma \preceq \cP_{l+1},\ldots,\cP_k$ for a given integer $l$ with $0 \leq l \leq k$ if and only if $\cS_{\rm loose}(\cH, {\rm lub}(\cP_1,\ldots,\cP_l)) \preceq \cS_{\rm lacy}(\cH, {\rm glb}(\cP_{l+1},\ldots,\cP_k))$.
\end{Corollary}
\begin{proof} 
($\Rightarrow$) 
If $\Sigma$ satisfies  $\cP_1,\ldots,\cP_l \preceq \Sigma \preceq \cP_{l+1},\ldots,\cP_k$ then 
$${\rm lub}(\cP_1,\ldots,\cP_l) \preceq \Sigma \preceq {\rm glb}(\cP_{l+1},\ldots,\cP_k).$$ If, in addition, $\Sigma\subseteq\cH$ then 
$\cS_{\rm loose}(\cH, {\rm lub}(\cP_1,\ldots,\cP_l)) \preceq \cS_{\rm lacy}(\cH, {\rm glb}(\cP_{l+1},\ldots,\cP_k))$ by 
Theorem~\ref{inBetween}.

($\Leftarrow$) Conversely, suppose that $$\cS_{\rm loose}(\cH, {\rm lub}(\cP_1,\ldots,\cP_l)) \preceq \cS_{\rm lacy}(\cH, {\rm glb}(\cP_{l+1},\ldots,\cP_k)).$$   By Theorem \ref{inBetween} there exists a set $\Sigma\subseteq\cH$ with ${\rm lub}(\cP_1,\ldots,\cP_l) \preceq \Sigma \preceq {\rm glb}(\cP_{l+1},\ldots,\cP_k)$.  Since  $\cP_1,\ldots,\cP_k\preceq{\rm lub}(\cP_1,\ldots,\cP_l)$ and ${\rm glb}(\cP_{l+1},\ldots,\cP_k)\preceq \cP_{l+1},\ldots,\cP_k$ it follows that $$\cP_1,\ldots,\cP_l \preceq \Sigma \preceq \cP_{l+1},\ldots,\cP_k.$$
\end{proof}
Note, that if $l=0$, the properties (AE) hold for $k$ phenotypic partitions, while if $l=k$, the properties (BE) hold for $k$ phenotypic partitions.

\section{Generating phenotypic partitions from edge-marked trees}

In this final section, we return to the setting of single partitions and ask when a phenotypic partition of $\x$ can be realized by marking some edges of a given phylogenetic $X$--tree $T=(V,E)$, and assigning two individuals to the same phenotype if there is no marked edge on the path between them. 

This process has a clear biological interpretation:  A  marked edge denotes that  a new phenotype evolves to  replace the existing phenotype, and it is assumed here that a phenotype that has already appeared will not appear a second time; in other words, the evolution  of the phenotype is homoplasy-free. Therefore, we define a {\em marking map}  $m:E\to\{0,1\}$ on $T$ as follows:
\begin{equation*}
m(e) =
\begin{cases}
1, & \text{if } e\in E\text{ is marked};\\
0, & \text{otherwise.}
\end{cases}
\end{equation*}
For $x,y\in\x$ we define the relation $\sim_m$ by $x\sim_m y$ if $m(e)=0$ for every edge $e$ on the path from $x$ to $y$. This induces a corresponding phenotypic partition $\cP_m$ as the set of equivalence classes of $\x$ under $\sim_m$. Not all possible partitions of $\x$ can be realized in this way; an  example for such a case is given in Fig. \ref{wrong}.
\begin{figure}[h]
	\centering
	\centering
\includegraphics[scale=1.0]{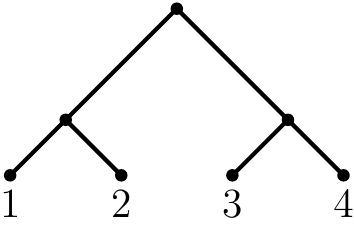}
	\caption{A rooted phylogenetic tree on which the phenotypic partition $\cP=\{\{1,3\},\{2,4\}\}$ cannot be realized by marking any subset of edges.}
	\label{wrong}
\end{figure}

In general, a partition $\cP$ of $\x$ can be realized by a marking map on a given rooted phylogenetic tree $T$  on $\x$ if and only if that partition is `convex' on $T$ (i.e. the minimal subtrees of $T$ that connect the leaves in each block of $\cP$ are vertex-disjoint), for details, see \cite{steel2016phylogeny}.
In general, there can be more than one marking map that leads to the same phenotypic partition, as Fig. \ref{two_markings} shows. For this reason, we define the relation $\simeq$ by $m\simeq m'$ if $\cP_m=\cP_{m'}$.

\begin{figure}[!h]
	\centering
\includegraphics[scale=1.0]{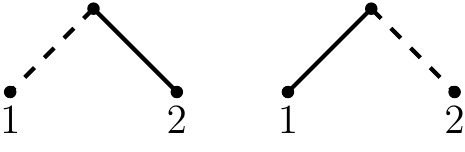}
	\caption{A rooted phylogenetic tree with different labeled edges which induce the phenotypic partition $\cP~=~\{\{1\},\{2\}\}$. The dashed line is marked and the solid line is not.}\label{two_markings}
\end{figure}

There is a close connection between the  property (E) for the phenotypic partition and its realization by marking maps, as we now describe.

\begin{Proposition}
\label{helpspro}
Suppose that $m$ is a marking map on a rooted phylogenetic tree $T$ on $\x$ with associated hierarchy $\cH$.  
	$\cP_m\subseteq\cH$ if and only if for any two distinct leaves $x,y$ with $x\sim_m y$, we have $m(e)=0$ for every edge $e$ in the subtree of $T$ with root $\text{lca}_T({x,y})$.
\end{Proposition}

\begin{proof}
\mbox{}
	For this proof, we introduce some further notation. Given a rooted phylogenetic tree $T$ and a vertex $v$ of $T$, let  \textit{subtree} $(T,v)$ denote the rooted phylogenetic tree which has root $v$ and contains all vertices and edges which are descendants from $v$ in $T$.  

	\begin{enumerate}
	\item[($\Rightarrow$)]
		Suppose that $\cP_m\subseteq\cH$ holds, and that $x,y$ are distinct
		leaves of $X$ with $x \sim_m y$. Let  $P$ be the block of $\cP_m$ containing $x$ and $y$. Then $P \in \cH$ (since $\cP_m\subseteq\cH$) and so $m(e)=0$ for each edge in the subtree ($T,\text{lca}_T(P))$. Since ${\rm lca}_T(P) \preceq_T {\rm lca}_T(\{x,y\})$ it follows that $m(e) = 0$ for every edge $e$ in the  $\text{subtree } (T,\text{lca}_T(\{x,y\}))$, as claimed.
		\item[$(\Leftarrow$)] Suppose that $\cP_m\nsubseteq\cH$ holds. Then there exists $P\in\cP_m$ with $P\notin\cH$. This means that there exists a leaf $x\in\x$ in the subtree ($T,\text{lca}_T(P)$) with $x\notin P$. Select $y\in P$. Hence, there exists an edge $e$ on the path from $x$ to $y$, with $m(e)=1$ and with $e$ present in  the $\text{subtree } (T,\text{lca}_T(\{x,y\}))$.
	\end{enumerate}
\end{proof}

We now state a necessary condition and a sufficient condition for a phenotypic partition to be realizable by a marking map on a given tree.

\begin{Proposition}
\label{propro}
Let $\cP$ be a phentypic partition of $\x$ and $T$ be a rooted phylogenetic tree on $\x$, with corresponding hierarchy $\cH$.
\begin{itemize}
\item[(i)]
If $\cP$ can be realized by a marking map on $T$ then at least one  block $P$ of $\cP$ is in $\cH$.
\item[(ii)]
	If $\cP\subseteq\cH$,  then the phenotypic partition $\cP$ can be realized by a marking map. 
\end{itemize}
\end{Proposition}

\begin{proof}
{\em Part (i):}
	Let $\cP$ be a phenotypic partition that can be realized by a marking map on $T$. If no edge is marked at all, then $\cP=\{\x\}$ and $\x\in\cH$ by definition.  If there exists at least one marked edge, then at least one marked edge $e$ has no marked edge descendant from it (i.e.  there exists a block $P\in\cP$ that consists of all descending leaves) and so   $P\in\cH$.\\
	{\em Part (ii):} 
	Suppose that $\cP\subseteq\cH$ holds. If $\cP = \{X\}$ then set $m(e)=0$ for each edge $e$ of $T$. Otherwise, if $\cP \neq \{X\}$,  for each block $P\in\cP$, let $e_P$ be the edge of $T$ that 
 is directed into the vertex $\text{lca}_T(P)$. Then set $m(e) = 1$, if $e = e_P$ for some $P \in\cP$ and set $m(e)=0$ otherwise. This gives a marking map that realizes $\cP$.
\end{proof}

\subsection{Maximum marking maps and the glb and the lub of multiple markings}

In order to extend the study to multiple marking maps in partitions,  we first require some further notation. 

Given a partition $\cP$ and a rooted phylogenetic tree $T$  we define a marking map $m_\cP$ on $T$ by:
	\begin{equation*}
	m_\cP(e) =
	\begin{cases}
		0, & \text{if } e\text{ is on a path from }x\text{ to }y, x\sim y \text{ for }\cP, x\neq y;\\
		1, & \text{otherwise.}
	\end{cases}
	\end{equation*}
Recall here that  $x\sim y$ for $\cP$ means that $x,y$ are in the same block of the partition $\cP$.
It can be  checked that the partition associated with $m_\cP$ (i.e. $\cP_{m'}$ for $m'= m_{\cP}$) is equal to or refines $\cP$ in the lattice of partitions of $\x$ (i.e. $\cP_{m'} \leq \cP$)  though $\cP_{m'}$ need not equal $\cP$ (for example, in Fig.~\ref{wrong}  with $\cP =  \{\{1,2\}, \{3,4\}\}$, the partition associated with $m_\cP$ is $\x = \{1,2,3,4\}$).

Notice that any marking map $m$ for a rooted tree $T=(V,E)$ is determined by the set $$\cE(m) := \{e \in E : m(e)=1\},$$
consisting of the edges of $T$ that are marked.

For a marking map $m$ on $T$, we let $\overline{m}$ be the marking map $m_\cP$ in which $\cP$ is taken to be $\cP_m$. 
Thus, $\overline{m}(e)$ takes the value 1, unless there is some pair $x,y \in X$ with $x \sim_m y$ for which $e$ lies on the path connecting $x$ and $y$. In particular,
$\overline{m}$ is the marking map that has the maximal number of edges assigned to state 1, while still inducing the same partition as $\cP_m$ (this is formalized
in part (ii) of Lemma~\ref{imp}).

\begin{Lemma}
\label{imp}
Let $m$ and $m'$ be marking maps on a rooted phylogenetic $T$ on $\x$ with associated hierarchy $\cH$.
\begin{itemize}
\item[(i)]
If  $\cE(m) \subseteq \cE(m')$, then $\cP_{m'} \preceq \cP_m.$
\item[(ii)] $m \simeq \overline{m}$ and 
if  $m' \simeq m$, then $\cE(m') \subseteq \cE(\overline{m})$.   Moreover, $\overline{(\overline{m})} = \overline{m}$.
\item[(iii)]
Suppose that $\cP_m \subseteq \cH$ and that $e, e'$ are edges of $T$ with $e'$ on the path from the root of $T$ to $e$. If $e \in \cE(\overline{m})$, then $e' \in \cE(\overline{m})$.
  \end{itemize}
\end{Lemma}
\begin{proof}
The proof of parts (i) and (ii) is  straightforward.  For part (iii), suppose that $e'$ is not in $\cE(\overline{m})$. Then, by definition of $\overline{m}$, there exist leaves $x, x' \in P \in \cP_m$ for which 
the path from $x$ to $x'$ includes $e'$.  By part (ii) of this lemma,  $\cP_m = \cP_{\overline{m}}$ and so, by Proposition~\ref{helpspro}, $\overline{m}(e'')=0$ for every edge $e''$ in the subtree of $T$ with root ${\rm lca}_T(x,x')$. However,  edge $e$ lies in this subtree, yet $\overline{m}(e)=1$ so $e \not\in \cE(\overline{m})$. 
\end{proof}

Given marking maps $m_1, \ldots, m_k$ on a tree $T$, we let
$m_1 \vee m_2 \vee \cdots \vee m_k$  and $m_1 \wedge m_2 \wedge \cdots \wedge m_k$ be the marking maps on $T$ defined by:
$$\cE(m_1 \vee m_2 \vee \cdots \vee m_k) =  \bigcup_{i=1}^k \cE(m_i),$$
and
$$ \cE(m_1 \wedge m_2 \wedge \cdots \wedge m_k) = \bigcap_{i=1}^k \cE(m_i).$$
In words, $m_1 \vee m_2 \vee \cdots \vee m_k$ is the marking map that assigns 1 to all edges of $T$ that are marked 1 by at least one $m_i$ ($1 \leq i \leq k$), while $m_1 \wedge m_2 \wedge \cdots \wedge m_k$ assigns 1 only to the edges
of $T$ that  are marked 1 by every $m_i$ ($1 \leq i \leq k$).

The following result describes the glb and lub of a collection of partitions (relevant to Theorem~\ref{more}) in the setting where  each partition is derived from a marking map on a given tree, in terms of a single marking map on that tree.

\begin{Theorem}
\label{long}
	Consider  a rooted phylogenetic tree $T=(V,E)$ with associated hierarchy $\cH$, and marking maps $m_i: E\to\{0,1\}$, for $i=1, \ldots, k$. Then
	\begin{itemize}
	\item[(ii)]${\rm glb}(\cP_{m_1},\cP_{m_2}, \ldots, \cP_{m_k})=\cP_{m_1 \vee m_2 \vee \cdots \vee m_k}.$
	\item[(ii)] 
	Provided that $\cP_{m_i} \subseteq \cH$ for each $1, \ldots, k$, we also have:
$${\rm lub}(\cP_{m_1},\cP_{m_2}, \ldots, \cP_{m_k})=\cP_{\overline{m_1} \wedge \overline{m_2} \wedge \cdots \wedge \overline{m_k}}.$$
\end{itemize}
\end{Theorem}

Before proceeding to the proof, we note that the additional assumption ($\cP_{m_i} \subseteq \cH$) in (ii) is required, even in the special case where $k=2$ and $m_i = \overline{m_i}$ for $i=1,2$.  
An example to demonstrate this is shown in Fig.~\ref{wrong2}.  In this example,  ${\rm lub}(\cP_{m_1},\cP_{m_2}, \ldots, \cP_{m_k})$ is not  realized by any marking map on $T$. 
Notice that Lemma~\ref{lemma} and  Proposition~\ref{propro}(ii) ensure that when $\cP_{m_i} \subseteq \cH$ for each $1, \ldots, k$, then ${\rm lub}(\cP_{m_1},\cP_{m_2}, \ldots, \cP_{m_k})$
is realized by at least  some marking map on $T$,  and Theorem~\ref{long}(ii) gives an explicit description of such a  map.
Note also, in part (ii) that we have replaced $m_i$ by $\overline{m_i}$ on the right-hand side. Without such a replacement, the identity in (ii) does not hold, as Fig.~\ref{two_markings} shows. 

\begin{figure}[h]
	\centering
	\centering
\includegraphics[scale=1.0]{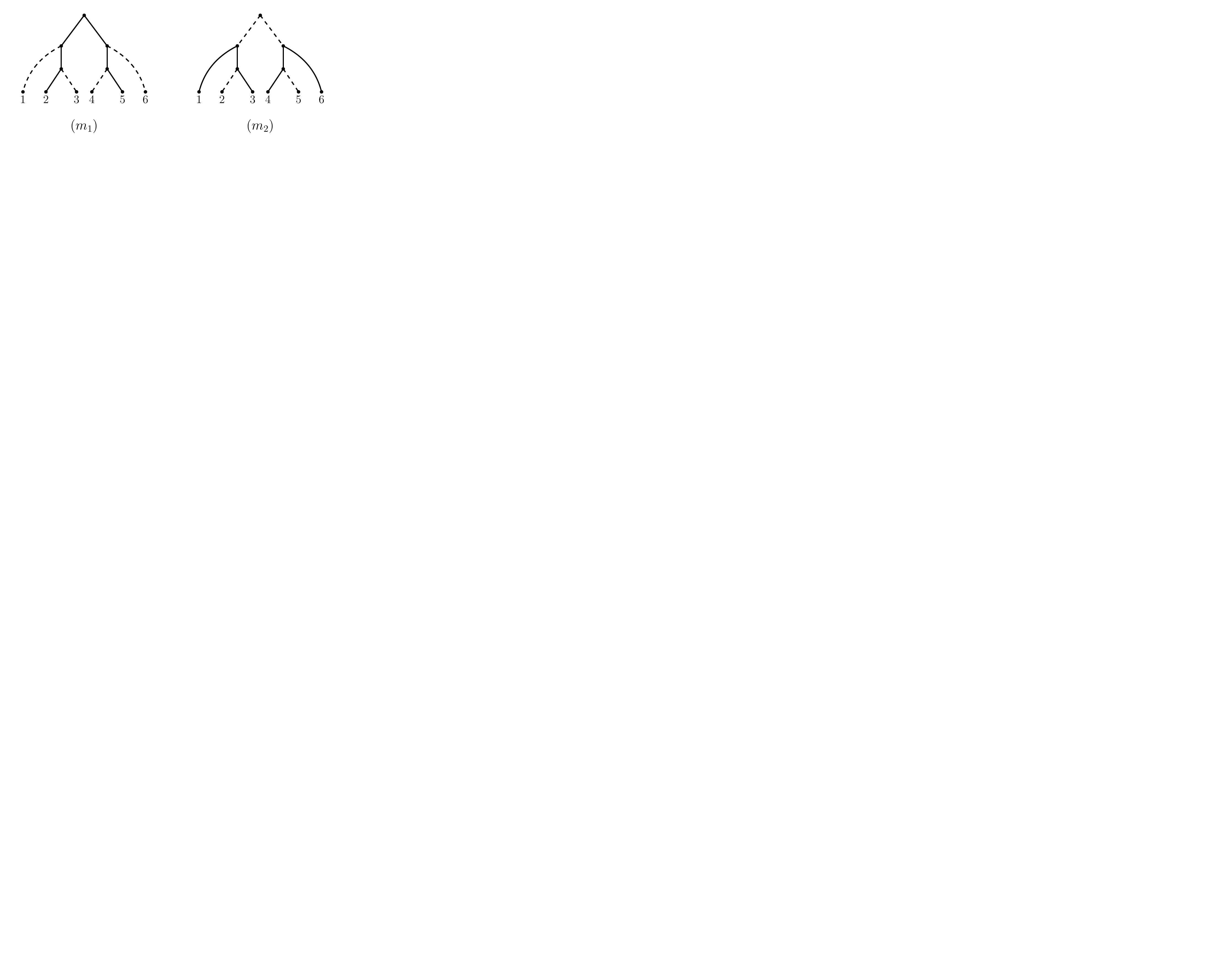}
	\caption{{\em Left:} The marking map $m_1$ (where a dashed edge $e$  corresponds to $m_1(e)=1$)  induces the phenotypic partition $\cP_{m_1}$ in which the only block of size $>1$ is $\{2,5\}$. {\em Right:} The marking map $m_2$ (where a dashed edge $e$  corresponds to  $m_2(e)=1$)  induces the phenotypic partition $\cP_{m_2}$ in which the only blocks of size $>1$ are $\{1,3\}$ and $\{4,6\}$.  Thus, the least upper bound (lub) of these two partitions contains the three blocks
	$\{1, 3\}, \{2,5\}, \{4,6\}$. However, the marking map $m_1 \wedge m_2$ assigns $0$ to all the edges of the tree shown, and so $\cP_{m_1 \wedge m_2} =X$.  Note that in this example, 
	$m_1 = \overline{m_1}$ and $m_2 = \overline{m_2}$.  }
	\label{wrong2}
\end{figure}

\bigskip

{\em Proof of Theorem~\ref{long}.}
\begin{proof}

{\em Part (i):}
First, we consider the case $k=2$ of just two marking maps $m_1$ and $m_2$ on $T$. 
Since $\cE(m_i) \subseteq \cE(m_1 \vee m_2)$,  for $i=1,2$, Lemma~\ref{imp}(i) gives
$\cP_{m_1 \vee m_2} \preceq \cP_{m_i}$ for $i=1,2$, and therefore:
$$\cP_{m_1 \vee m_2} \preceq {\rm glb} (\cP_{m_1}, \cP_{m_2}).$$
To replace $\preceq$ with equality it remains to show that ${\rm glb} (\cP_{m_1}, \cP_{m_2}) \preceq \cP_{m_1 \vee m_2}$, which means that if $S \in {\rm glb} (\cP_{m_1}, \cP_{m_2})$ and $S' \in \cP_{m_1 \vee m_2}$
with $S \cap S' \neq \emptyset$ then $S\subseteq S'$.  Recall (Section \ref{partitions}) that $S$ is a block of the glb of two given partitions in the poset $\p(\x)$ precisely if $S$ is the (non-empty) intersection of a block from the first partition with a block from the second. Thus, we may suppose that $S=\beta_1 \cap \beta_2$ where $\beta_1 \in \cP_{m_1}$ and $\beta_2 \in \cP_{m_2}$. Since $S \cap S' \neq \emptyset$ select $x \in S \cap S'$. Then $S'$ is the set of leaves  of $T$ whose path to  $x$ does not cross an edge marked 1 by either $m_1$ or $m_2$ (or both).  On the other hand, 
$\beta_i$ is the set of leaves of $T$ whose path to $x$ does not cross an edge marked 1 by $m_i$.  Thus, if  $y \in \beta_1 \cap \beta_2 = S$, then $y \in S'$, and so $S \subseteq S'$ as claimed. 

The case where $k>2$ now follows by
the associativity of $\vee$ on the lattice $(\p(\x),\preceq)$. We have
$$ {\rm glb}(\cP_{m_1},\cP_{m_2}, \ldots, \cP_{m_k})=  {\rm glb}(\cP_{m_1}, {\rm glb}(\cP_{m_2}, \ldots, \cP_{m_k})),$$ and so part (i) follows by induction, from 
the case where $k=2$.  This establishes part (i) of Theorem~\ref{long}.

{\em Part (ii):}
Again, we first  consider the case of $k=2$ marking maps $m_1$ and $m_2$ on $T$. 
Note that $\cE(\overline{m_1}\wedge \overline{m_2}) \subseteq \cE(\overline{m_i})$ for  $i=1, 2$.
So Lemma~\ref{imp} gives
$\cP_{m_i} = \cP_{\overline{m_i}}   \preceq \cP_{\overline{m_1} \wedge \overline{m_2}}$ for  $i=1,2$ and therefore:
 $${\rm lub}(\cP_{m_1}, \cP_{m_2}) \preceq \cP_{\overline{m_1} \wedge \overline{m_2}}.$$
 
To replace $\preceq$ with equality it remains to show that $\cP_{\overline{m_1} \wedge \overline{m_2}} \preceq {\rm lub} (\cP_{m_1}, \cP_{m_2})$, which means that if $S \in\cP_{\overline{m_1} \wedge \overline{m_2}}$ and $S' \in {\rm lub} (\cP_{m_1}, \cP_{m_2})$
with $S \cap S' \neq \emptyset$ then $S\subseteq S'$.  To this end, suppose that $x \in S \cap S'$ and that $x\sim_{\overline{m_1} \wedge {\overline{m_2}}} y$.
We will establish the following:

\begin{itemize}

\item[] {\bf Claim 1:} Given  $\cP_{m_1}, \cP_{m_2} \subseteq \cH$,  if $x \sim_{\overline{m_1} \wedge {\overline{m_2}}} y$, then  either
 $x \sim_{\overline{m_1}} y$ or $x \sim_{\overline{m_2}} y$ (or both) holds.
\end{itemize}

It follows from Claim 1 that  $x \approx y$ (where $\approx$ is the equivalence relation in the definition of lub from Section~\ref{partitions}) and so $y  \in S' \in {\rm lub} (\cP_{m_1}, \cP_{m_2})$, as required to show that 
$\cP_{\overline{m_1} \wedge \overline{m_2}} \preceq {\rm lub} (\cP_{m_1}, \cP_{m_2})$, and thereby to establish part (ii) in the case $k=2$.

\bigskip

Thus, for the $k=2$ case, it remains to prove Claim 1. We do this by assuming that Claim 1 is false, and derive a contradiction.
Now, if Claim 1 is false, then  the path from $x$ to $y$ crosses at least one edge $e_1$ with $\overline{m_1}(e_1)=1$ and at least one edge $e_2$ with
$\overline{m_2}(e_2)=1$, and there is no edge $e$ on this path with $\overline{m_1}(e)=\overline{m_2}(e)=1$ (since $x\sim_{\overline{m_1} \wedge {\overline{m_2}}} y$). 

Let $v= {\rm lca}_T(\{x, y\})$ and let  $v_x$ (respectively, $v_y$) be vertex adjacent to $v$ that is on the path from $v$ to $x$ (respectively, to $y$).
By Lemma~\ref{imp}(iii), it follows  that the edge $e=(v, v_x)$ has $\overline{m_i}(e)=1$ and $\overline{m_j}(e)=0$, while the edge  $e'=(v, v_y)$ has $\overline{m_i}(e)=0$ and $\overline{m_j}(e)=1$, 
 where $\{i,j\}=\{1,2\}$.  Without loss of generality, we may assume that $i=1$ and $j=2$.  The condition  that $\overline{m_2}(e)=0$  then implies (by the definition of $\overline{m_2}$) that 
 $e$ lies on a path connecting two leaves  of $T$ (which lie in the same block of $\cP_{m_2} = \cP_{\overline{m_2}}$)  and every edge in that path has $\overline{m_2}$ value equal to  zero. 
 One of these two leaves is in the subtree of $T$ below $v_x$, the other leaf $x'$ must also lie below $v$ since the edge $e''$ of $T$ that ends in $v$ has $\overline{m_1}(e'')=\overline{m_2}(e'')=1$ by Lemma~\ref{imp}(iii). Moreover, $x'$ cannot lie in the subtree
 with root $v_y$ (since the path would then cross $e'$ which has $\overline{m_2}(e')=1$), nor can it be in the subtree with root $v_x$ (since the path from $x'$ to $x$ has to cross $e$).
 Thus, $x'$ must lie in a third pendant subtree that attaches below $v$, as indicated in Fig.~\ref{proofig}.

 \begin{figure}[h]
	\centering
	\centering
	\includegraphics[scale=1.0]{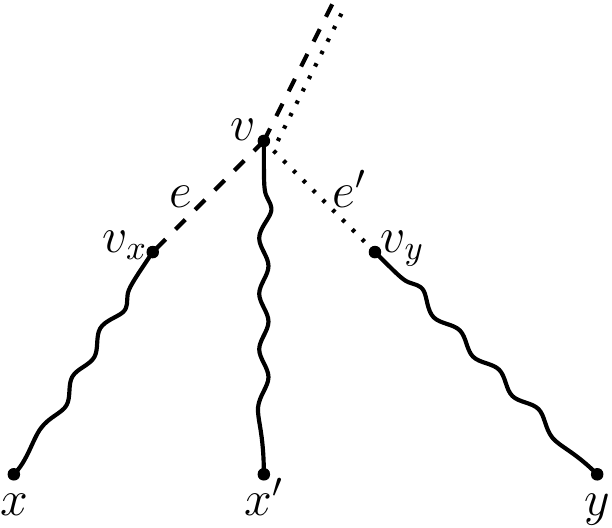}
	\caption{For the two leaves $x,y$ for which $x \sim_{\overline{m_1} \wedge {\overline{m_2}}} y$ the edge $e$ incident with $v={\rm lca}_T(\{x,y\})$ has
	$\overline{m_1}(e)=1$ and $\overline{m_2}(e)=0$ while the edge $e'$ has $\overline{m_1}(e')=0$ and $\overline{m_2}(e')=1$ (each  edge above $v$ has an $\overline{m_1}$ and
	$\overline{m_2}$ value equal to 1). Every edge on the path from $v$ to $x$, and on the path from $v$ to $x'$  has $\overline{m_2}$ value equal to 0, and so $x\sim_{\overline{m_2}} x'$.}
	\label{proofig}
\end{figure}

 But then Proposition~\ref{helpspro} implies that  $\cP_{m_2}= \cP_{\overline{m_2}}$ is not a subset of $\cH$, since $x \sim_{\overline{m_2}} x'$ but $\overline{m_2}$ is not zero on every 
 edge in the subtree $(T, {\rm lca}_T(\{x,x'\}))$ since this subtree included the edge $e'$ for which $\overline{m_2}(e')=1$.  This contradiction establishes Claim 1, and thereby the proof
 in the case where $k=2$.

\bigskip

The case where $k>2$ now follows by
the associativity of $\wedge$ on the lattice $(\p(\x),\preceq)$ and the fact that $\overline{(\overline{m_i})} = \overline{m_i}.$ We have
$$ {\rm lub}(\cP_{m_1},\cP_{m_2}, \ldots, \cP_{m_k})=  {\rm lub}(\cP_{m_1}, {\rm lub}(\cP_{m_2}, \ldots, \cP_{m_k})),$$ and so part (ii) follows by induction, from 
the case where $k=2$.  This establishes part (ii) of Theorem~\ref{long}.
\end{proof}

\section*{Concluding comments}

Our purpose in this paper was two-fold. First, we have represented some of the main results of \cite{manceau2016species} within the framework of lattice theory, which has allowed for more concise proofs, and elucidated the key combinatorial principles. Second, we have extended these results to the more general setting of $k\geq 2$ phenotypic partitions, and we have also explored the combinatorial aspects of representing phenotypic partitions  by marking maps.

An interesting project for future development could be to investigate the loose and lacy species partitions arising when the phenotypic character(s) evolve on a tree $T$ with leaf set $X$, and $\cH$ is the hierarchy associated with $T$.  The evolution of phenotypic characters on $T$ is typically modelled by 
discrete-state, continuous-time Markov processes that are already widely used in molecular phylogenetics \citep{fels}.  
Finite state Markov models can lead to homoplasy (convergent or reverse changes), while infinite state models always produce homoplasy-free character states at the leaves. Under either model, the character states at the leaves then induce a phenotypic partition of $X$, and  the stochastic properties of the  loose and lacy species  partition could then be studied. 

A first step would be to simply calculate  the probability that the 
loose and lacy partitions species coincide (i.e. a species partition satisfying all three properties (ABE), or, equivalently, each phenotype corresponds to a cluster of $\cH$). 
It would also be of interest to predict the distribution of block sizes in species partitions, a topic that has  long been of interest to biologists, dating back to \cite{yule} (see, for example, \cite{sco}). 

Finally, the definition and computation of  loose and lacy species partitions based on phylogenetic networks (rather than phylogenetic trees) could be a further interesting direction for future work.

\bigskip

	\section*{Acknowledgement}
	We thank Amaury Lambert for helpful discussions, Mareike Fischer for comments on an earlier version of this paper, and the (former) Allan Wilson Centre for funding this work.


\bibliographystyle{model2-names}

	\bibliography{references}

\begin{thebibliography}{18}
\expandafter\ifx\csname natexlab\endcsname\relax\def\natexlab#1{#1}\fi
\providecommand{\url}[1]{\texttt{#1}}
\providecommand{\href}[2]{#2}
\providecommand{\path}[1]{#1}
\providecommand{\DOIprefix}{doi:}
\providecommand{\ArXivprefix}{arXiv:}
\providecommand{\URLprefix}{URL: }
\providecommand{\Pubmedprefix}{pmid:}
\providecommand{\doi}[1]{\href{http://dx.doi.org/#1}{\path{#1}}}
\providecommand{\Pubmed}[1]{\href{pmid:#1}{\path{#1}}}
\providecommand{\bibinfo}[2]{#2}
\ifx\xfnm\relax \def\xfnm[#1]{\unskip,\space#1}\fi
\bibitem[{Aldous et~al.(2008)Aldous, Krikun and Popovic}]{ald}
\bibinfo{author}{Aldous, D.}, \bibinfo{author}{Krikun, M.},
  \bibinfo{author}{Popovic, L.}, \bibinfo{year}{2008}.
\newblock \bibinfo{title}{Stochastic models for phylogenetic trees on
  higher-order taxa.}
\newblock \bibinfo{journal}{J. Math. Biol.} \bibinfo{volume}{56},
  \bibinfo{pages}{525--557}.
\bibitem[{Alexander(2013)}]{ale}
\bibinfo{author}{Alexander, S.A.}, \bibinfo{year}{2013}.
\newblock \bibinfo{title}{Infinite graphs in systematic biology, with an
  application to the species problem}.
\newblock \bibinfo{journal}{Acta Biotheoretica} \bibinfo{volume}{61},
  \bibinfo{pages}{181--201}.
\bibitem[{Baum(1992)}]{bau}
\bibinfo{author}{Baum, D.}, \bibinfo{year}{1992}.
\newblock \bibinfo{title}{Phylogenetic species concepts}.
\newblock \bibinfo{journal}{Trends in Ecology and Evolution}
  \bibinfo{volume}{7}, \bibinfo{pages}{1--2}.
\bibitem[{Baum and Donoghue(1995)}]{bau2}
\bibinfo{author}{Baum, D.A.}, \bibinfo{author}{Donoghue, M.J.},
  \bibinfo{year}{1995}.
\newblock \bibinfo{title}{Choosing among alternative ``phylogenetic'' species
  concepts}.
\newblock \bibinfo{journal}{Systematic Biology} \bibinfo{volume}{20},
  \bibinfo{pages}{560--573}.
\bibitem[{B{\'o}na(2011)}]{bona2011walk}
\bibinfo{author}{B{\'o}na, M.}, \bibinfo{year}{2011}.
\newblock \bibinfo{title}{A {W}alk {T}hrough {C}ombinatorics: {A}n
  {I}ntroduction to {E}numeration and {G}raph {T}heory}.
\newblock \bibinfo{publisher}{World scientific}.
\bibitem[{De~Queiroz and Donoghue(1988)}]{deq}
\bibinfo{author}{De~Queiroz, K.}, \bibinfo{author}{Donoghue, M.J.},
  \bibinfo{year}{1988}.
\newblock \bibinfo{title}{Phylogenetic systematics and the species problem}.
\newblock \bibinfo{journal}{Cladistics} \bibinfo{volume}{4},
  \bibinfo{pages}{317--338}.
\bibitem[{Dress et~al.(2010)Dress, Moulton, Steel and Wu}]{dre}
\bibinfo{author}{Dress, A.}, \bibinfo{author}{Moulton, V.},
  \bibinfo{author}{Steel, M.}, \bibinfo{author}{Wu, T.}, \bibinfo{year}{2010}.
\newblock \bibinfo{title}{Species, clusters and the `{T}ree of {L}ife': A
  graph-theoretic perspective}.
\newblock \bibinfo{journal}{J. Theor. Biol.} \bibinfo{volume}{265},
  \bibinfo{pages}{535--542}.
\bibitem[{Dress et~al.(2011)Dress, Huber, Koolen, Moulton and
  Spillner}]{drebook}
\bibinfo{author}{Dress, A.W.M.}, \bibinfo{author}{Huber, K.T.},
  \bibinfo{author}{Koolen, J.}, \bibinfo{author}{Moulton, V.},
  \bibinfo{author}{Spillner, A.}, \bibinfo{year}{2011}.
\newblock \bibinfo{title}{Basic phylogenetic combinatorics}.
\newblock \bibinfo{publisher}{Cambridge University Press}.
\bibitem[{Felsenstein(2004)}]{fels}
\bibinfo{author}{Felsenstein, J.}, \bibinfo{year}{2004}.
\newblock \bibinfo{title}{Inferring phylogenies}.
\newblock \bibinfo{publisher}{Sinauer Press}.
\bibitem[{Holland et~al.(2010)Holland, Spencer, Worthy and Kennedy}]{hol}
\bibinfo{author}{Holland, B.R.}, \bibinfo{author}{Spencer, H.G.},
  \bibinfo{author}{Worthy, T.H.}, \bibinfo{author}{Kennedy, M.},
  \bibinfo{year}{2010}.
\newblock \bibinfo{title}{Identifying cliques of convergent characters:
  {C}oncerted evolution in the cormorants and shags}.
\newblock \bibinfo{journal}{Systematic Biology} \bibinfo{volume}{59},
  \bibinfo{pages}{433--445}.
\bibitem[{Kornet et~al.(1995)Kornet, Metz and Schellinx}]{kor}
\bibinfo{author}{Kornet, D.}, \bibinfo{author}{Metz, J.},
  \bibinfo{author}{Schellinx, H.}, \bibinfo{year}{1995}.
\newblock \bibinfo{title}{Internodons as equivalence classes in genealogical
  networks: {B}uilding-blocks for a rigorous species concept}.
\newblock \bibinfo{journal}{J. Math. Biol.} \bibinfo{volume}{34},
  \bibinfo{pages}{110--122}.
\bibitem[{Kwok(2011)}]{kwo}
\bibinfo{author}{Kwok, R.B.H.}, \bibinfo{year}{2011}.
\newblock \bibinfo{title}{Phylogeny, genealogy and the {L}innaean hierarchy:
  {A} logical analysis.}
\newblock \bibinfo{journal}{J. Math. Biol.} \bibinfo{volume}{63},
  \bibinfo{pages}{73--108}.
\bibitem[{Manceau and Lambert(2016)}]{manceau2016species}
\bibinfo{author}{Manceau, M.}, \bibinfo{author}{Lambert, A.},
  \bibinfo{year}{2016}.
\newblock \bibinfo{title}{The species problem from the modeler's point of
  view}.
\newblock \bibinfo{journal}{bioRxiv, 075580, https://doi.org/10.1101/075580} .
\bibitem[{Rosenberg(2007)}]{noa}
\bibinfo{author}{Rosenberg, N.A.}, \bibinfo{year}{2007}.
\newblock \bibinfo{title}{Statistical tests for taxonomic distinctiveness from
  observations of monophyly}.
\newblock \bibinfo{journal}{Evolution} \bibinfo{volume}{61},
  \bibinfo{pages}{317--323}.
\bibitem[{Scotland and Sanderson(2004)}]{sco}
\bibinfo{author}{Scotland, R.W.}, \bibinfo{author}{Sanderson, M.J.},
  \bibinfo{year}{2004}.
\newblock \bibinfo{title}{The significance of few versus many in the tree of
  life}.
\newblock \bibinfo{journal}{Science} \bibinfo{volume}{303},
  \bibinfo{pages}{643}.
\bibitem[{Steel(2016)}]{steel2016phylogeny}
\bibinfo{author}{Steel, M.}, \bibinfo{year}{2016}.
\newblock \bibinfo{title}{Phylogeny: Discrete and {R}andom {P}rocesses in
  {E}volution}.
\newblock \bibinfo{publisher}{SIAM}.
\bibitem[{Wheeler and Meier(2000)}]{que}
\bibinfo{author}{Wheeler, Q.D.}, \bibinfo{author}{Meier, R.},
  \bibinfo{year}{2000}.
\newblock \bibinfo{title}{Species Concepts and Phylogenetic Theory: A Debate}.
\newblock \bibinfo{publisher}{Colombia University Press, New York}.
\bibitem[{Yule(1925)}]{yule}
\bibinfo{author}{Yule, G.U.}, \bibinfo{year}{1925}.
\newblock \bibinfo{title}{A mathematical theory of evolution: Based on the
  conclusions of {D}r. {J}. {C}. {Willis}, {F.R.S.}}
\newblock \bibinfo{journal}{Philosophical Transactions of the Royal Soceity of
  London B.} \bibinfo{volume}{213}, \bibinfo{pages}{21--87}.

\end{thebibliography}

	\section{Appendix: Proof of Proposition~\ref{construct}.}
	Let $\widetilde{\cS}_{\rm loose}$ be the set of maximal elements  of $\cH_1$ and let $\widetilde{\cS}_{\rm lacy}$ be set of  maximal elements of $\cH_2$. 
We show that $\widetilde{\cS}_{\rm loose} = \cS_{\rm loose}$ and $\widetilde{\cS}_{\rm lacy} =\cS_{\rm lacy}$, by establishing the following properties:
	\begin{enumerate}
		\item[(i)] $\widetilde{\cS}_{\rm loose}$ and $\widetilde{\cS}_{\rm lacy}$ are partitions of $\x$,
		\item[(ii)] $\widetilde{\cS}_{\rm loose}$ satisfies (AE) and $\widetilde{\cS}_{\rm lacy}$ satisfies (BE),
		\item[(iii)] $\widetilde{\cS}_{\rm loose}$ is the finest partition of $\x$ that satisfies (AE) and $\widetilde{\cS}_{\rm lacy}$ is the coarsest partition of $\x$ that satisfies (BE).
	\end{enumerate}

\begin{proof}
{\em Part (i)}: 
		\begin{enumerate}
			\item 
			Proof that $\widetilde{\cS}_{\rm loose}$ is a partition of $X$. \\
			Suppose that $x \in X$.  Since $\cP$ is a partition of $X$, there is a set $P \in \cP$ with $x\in P$, so $x \in h_P$  (where $h_P$ is as in Eqn.~(\ref{hpeq})). Thus, $x$ is contained in at least one maximal element of $\cH_1$ (and so in 
		some set of  $\widetilde{\cS}_{\rm loose}$).  Moreover,  two different maximal elements of $\cH_1$, say $h$ and $h'$, have
		empty intersection. For otherwise, since  $h, h' \in \cH$, the nesting property of hierarchies implies that either $h \subsetneq h'$ or $h' \subsetneq h$ which is impossible if both sets are maximal.
	
			\item Proof that $\widetilde{\cS}_{\rm lacy}$ is a partition of $X$. \\
			Suppose that $x \in X$.  Then $x \in P$ for some set $P \in \cP$, and since $h=\{x\} \in \cH$, and $h \in \cP$,
			$\cH_2$ has a maximal set $h'$ containing $x$, so $x \in h' \in \cH_2$.  Moreover, two different maximal elements of $\cH_2$
			have empty intersection, for the same reason as holds for $\cH_1$.
		\end{enumerate}
{\em Part (ii)}:
			\begin{enumerate}
				\item Proof that $\widetilde{\cS}_{\rm loose}$ satisfies (AE):\\
			For any $h \in \widetilde{\cS}_{\rm loose}$ we have $h=h_P$ for some $P \in \cP$. Suppose that $Q \in \cP$ has $Q \cap h \neq \emptyset$.  Then $h_{Q} \cap h_{P} \neq \emptyset$ and so 
			$h_Q \subseteq h_P$ or $h_P \subseteq h_Q$ (since $h_{Q}$ and $h_{P}$ are both elements of a hierarchy). The maximality condition in the definition of $\widetilde{\cS}_{\rm loose}$ applied to $h=h_P$ now ensures
			that $h_Q \subseteq h_P$, so $Q\subseteq h_Q \subseteq h_P=h$ and so property (A) holds for $\widetilde{\cS}_{\rm loose}$; moreover, since $h \in \cH$, property (E) holds also. 
				
				\item Proof that $\widetilde{\cS}_{\rm lacy}$ satisfies (BE):\\
			For any $h \in 	\widetilde{\cS}_{\rm lacy}$ there is a set $P \in \cP$ with $h \subseteq P$.  For any $Q \in \cP$ with $Q \neq P$ we must have  $Q \cap h = \emptyset$, since otherwise $P \cap Q \neq \emptyset$. Thus, property (B) holds for $\widetilde{\cS}_{\rm loose}$; moreover, since $h \in \cH$, property (E) holds also. 			
			\end{enumerate}			
{\em Part (iii)}:					
				\begin{enumerate}
					\item Proof that $\widetilde{\cS}_{\rm loose}$ is the finest species partition satisfying (AE):\\
					This means for $\cP\preceq\cW \preceq\widetilde{\cS}_{\rm loose}$ with $\cW\subseteq\cH$ it holds that $\cW=\widetilde{\cS}_{\rm loose}$.\\
Suppose this is not the case (we will derive a contradiction). Then there must exist a set  $h_P \in \widetilde{\cS}_{\rm loose}$ 
for which $h_P$ is the disjoint union of $k \geq 2$ non-empty sets  $h_1, h_2, \ldots, h_k$ in $\cW$. Since $\cP \preceq \cW$, if $P \cap h_i \neq \emptyset$ we have $P \subseteq h_i$ and so $h_P \subseteq h_i$, which implies that $h_P = h_i$. Thus, $P$ can intersect at most one of the sets $h_i$, and so it must intersect exactly one set, say $h_1$, since the union of $h_1,\ldots, h_k$ is $h_P$ which contains $P$.   But this implies that each of the remaining $k-1$ sets is empty (since the union is disjoint), contradicting the requirement that the sets $h_i$ are non-empty.

					\item Proof that $\widetilde{\cS}_{\rm lacy}$ is the coarsest species partition satisfying (BE):\\
					This means for $\widetilde{\cS}_{\rm lacy}\preceq \cV\preceq\cP$ with $\cV \subseteq\cH$ it holds that $\cV=\widetilde{\cS}_{\rm lacy}$.\\
Suppose this is not the case (we will derive a contradiction). Then there must exist some element $h \in \cV$ and some element $h' \in \widetilde{\cS}_{\rm lacy}$ with $h' \subsetneq h$.  Let $P$ be an element of $\cP$ with $h \subseteq P$ (such an element must exist since
$h \in \cV \preceq\cP$).  Then $h$ and  $h'$ are both elements of $\cH_2$ (since they are both elements of $\cH$ and are contained within $P$) and since $h' \in \widetilde{\cS}_{\rm lacy}$ is a maximal element of $\cH_2$ we cannot have $h'\subsetneq h$, which provides the required contradiction.
				\end{enumerate}
\end{proof}

\end{document}